\documentclass[lettersize, journal]{IEEEtran}
\usepackage{cite}
\usepackage{pgfplots}

\pgfplotsset{compat=1.10}

\usepgfplotslibrary{fillbetween}

\usepackage{bm}
\usepackage{lipsum}
\usepackage{makeidx}
\usepackage{enumerate}
\usepackage{color}
\usepackage{cite}
\usepackage{amsmath,amsthm}   
\usepackage{amssymb}
\usepackage{nomencl}
\usepackage{multirow}
\usepackage{graphicx}
\usepackage{epstopdf}
\usepackage{threeparttable}
\usepackage{multicol}
\DeclareGraphicsExtensions{.pdf,.jpeg,.png,.jpg,.emf,.eps}
\usepackage{calc}
\usepackage{here}
\usepackage{url}

\usepackage{upref}
\usepackage{comment}
\usepackage{times}
\usepackage{dsfont}
\usepackage{epic,eepic}
\usepackage{rawfonts}
\usepackage[T1]{fontenc}
\usepackage{latexsym}
\usepackage{amsfonts}

\usepackage{capitalgreekitalic}

\usepackage[font=small,labelfont=bf]{caption}
\usepackage[font=small,labelfont=bf]{subcaption}

\usepackage{xcolor}
\usepackage{tikz,pgfplots}
\usetikzlibrary{calc}
\usetikzlibrary{intersections}
\usetikzlibrary{arrows,shapes}
\usetikzlibrary{positioning}

\newtheorem{lemma}{Lemma}

\newtheorem{remark}{Remark}

\theoremstyle{definition}

\allowdisplaybreaks

\begin{document}
\title{Rate Region of RIS-Aided URLLC Broadcast Channels: Diagonal versus Beyond Diagonal Globally Passive RIS
}
\author{Mohammad Soleymani, \emph{Member, IEEE},  
Alessio Zappone, \emph{Senior Member, IEEE}, 
Eduard Jorswieck, \emph{Fellow, IEEE}, Marco Di Renzo \emph{Fellow, IEEE}, and 
Ignacio Santamaria, \emph{Senior Member, IEEE}, 
 \\ \thanks{ 
Mohammad Soleymani is with the Signal \& System Theory Group, Universit\"at Paderborn, Germany (email: \protect\url{mohammad.soleymani@uni-paderborn.de}).  

Alessio Zappone is with the Department of Electrical and Information
Engineering, University of Cassino and Southern Lazio, 03043 Cassino, Italy,
and also with CNIT, 43124 Parma, Italy (e-mail: \protect\url{alessio.zappone@unicas.it)}).

Eduard Jorswieck is with the Institute for Communications Technology, Technische Universit\"at Braunschweig, 38106 Braunschweig, Germany
(email: \protect\url{jorswieck@ifn.ing.tu-bs.de}).

Marco Di Renzo is with Université Paris-Saclay, CNRS, CentraleSupélec, Laboratoire des Signaux et Systèmes, 91192 Gif-sur-Yvette, France (email: \protect\url{marco.di-renzo@universite-paris-saclay.fr}).

Ignacio Santamaria is with the Department of Communications Engineering, Universidad de Cantabria, Spain (email: \protect\url{i.santamaria@unican.es}).

The work of A. Zappone was funded by the European Union - NextGenerationEU under the project NRRP RESTART, RESearch and innovation on future Telecommunications systems and networks, to make Italy more smART PE 00000001 (CUP D43C22003080001) [MUR Decree n. 341- 15/03/2022] - Cascade Call launched by SPOKE 3 POLIMI: “SPARKS” project, and project PRIN GARDEN (CUP H53D23000480001).
The work of E. Jorswieck was supported by the Federal Ministry of Education and Research (BMBF, Germany) through the Program of Souver\"an. Digital. Vernetzt. joint Project 6G-RIC, under Grant 16KISK031, and by European Union's (EU's) Horizon Europe project 6G-SENSES under Grant 101139282. The work of M. Di Renzo was supported in part by the European Commission through the Horizon Europe project titled COVER under Grant 101086228, the Horizon Europe project titled UNITE under Grant 101129618, and the Horizon Europe project titled INSTINCT under Grant 101139161, as well as by the Agence Nationale de la Recherche (ANR) through the France 2030 project titled ANR-PEPR Networks of the Future under Grant NF-PERSEUS 22-PEFT-004, and by the CHIST-ERA project titled PASSIONATE under Grants CHIST-ERA-22-WAI-04 and ANR-23-CHR4-0003-01. The work of I. Santamaria was funded by MCIN/AEI 10.13039/501100011033, under Grant PID2022-137099NB-C43 (MADDIE), and by the Horizon Europe project 6G-SENSES under Grant 101139282. 
}}
\maketitle
\begin{abstract}
 We analyze the finite-block-length rate region of wireless systems aided by reconfigurable intelligent surfaces (RISs), employing treating interference as noise. We consider three nearly passive RIS architectures, including locally passive (LP) diagonal (D), globally passive (GP) D, and GP beyond diagonal (BD) RISs. In a GP RIS, the power constraint is applied globally to the whole surface, while some elements may amplify the incident signal locally. The considered RIS architectures provide substantial performance gains compared with systems operating without RIS. GP BD-RIS outperforms, at the price of increasing the complexity, LP and GP D-RIS as it enlarges the feasible set of allowed solutions. However, the gain provided by BD-RIS decreases with the number of RIS elements. Additionally, deploying RISs provides higher gains as the reliability/latency requirement becomes more stringent.
\end{abstract} 

\section{Introduction}\label{1}
The reconfigurable intelligent surface (RIS) is a promising technology to improve the spectral and energy efficiencies of wireless communication systems \cite{di2020smart, wu2021intelligent}.  {A key feature of RIS is to operate in a nearly passive manner, without requiring power amplifiers, thus potentially reducing the operational costs, and leading to more environmentally friendly and sustainable communication systems.} 

 {A nearly passive RIS can operate in two possible modes \cite{fotock2023energy}: (i) each RIS element operates in a nearly passive mode, hence the amplitude of the reflection coefficient of each RIS element is no greater than one. This architecture is referred to as locally passive (LP); (ii) the total output power of the RIS is no greater than the total input power. In this case, some RIS elements may amplify the incident signal, while the others attenuate it. Globally, however, the whole RIS does not amplify the incident signal. This architecture is referred to as globally passive (GP).}

As far as the reconfigurability of the RIS is concerned, the typical architecture assumes that each RIS element is programmed via independent electronic circuits. This architecture is referred to as diagonal (D) RIS.  {To further increase the number of optimization variables and hence the wave-domain processing capabilities, the RIS elements may be programmed via a fully-connected network of electronic circuits \cite{li2023beyond}. This architecture is referred to as beyond-diagonal (BD) RIS. A nearly passive BD-RIS is Pareto-optimal from the information-theoretic point of view \cite{BartoliADDR23}. The price to pay is, however, a higher implementation complexity and a higher power consumption, compared with D-RIS \cite{soleymani2024energy}.}

A key requirement for sixth-generation (6G) networks is to enhance the link reliability and to reduce the transmission latency \cite{wang2023road}. RIS has been shown to be a suitable technology for increasing the reliability and, at the same time, reducing the latency of wireless systems \cite{soleymani2023spectral, almekhlafi2021joint, soleymani2023optimization, vu2022intelligent,  soleymani2024optimization}. For instance, the authors of \cite{soleymani2023spectral} have shown that a simultaneously transmitting and reflecting (STAR) RIS improves the latency and reliability in multi-cell broadcast channels (BCs). The authors of \cite{soleymani2023optimization} have proposed an optimization framework for 1-layer rate splitting multiple access in RIS-aided multiple-input single-output (MISO) ultra-reliable and low latency communication (URLLC) systems. The authors of \cite{vu2022intelligent} have shown that RIS improves the performance of a two-user single-antenna BC with non-orthogonal multiple access (NOMA).  {No other papers in the literature have, however, considered GP RISs for URLLC applications.}

In this paper, we aim at analyzing the achievable rate region of an RIS-aided MISO BC, assuming a finite block length (FBL) coding and treating interference as noise (TIN) as decoding strategy. We devise an algorithm for optimizing the rate region for FBL coding, by establishing a link between the rate region and the signal-to-interference-plus-noise ratio (SINR) region, which can be applied if the FBL rate increases monotonically with the SINR. In addition, we compare the average max-min rates of three RIS architectures: LP D-RIS, GP D-RIS, and GP BD-RIS.  {We show that GP BD-RIS outperforms, in general, the other RIS architectures, at the price of increasing the hardware and computational complexities. In addition, the performance gain provided by BD-RIS is shown to decrease as the number of RIS elements increases.}

\textit{Notations:} The trace of a square matrix ${\bf X}$ is denoted as  $\text{Tr}({\bf X})$. The complex Gaussian random variable $x$ with mean $\mu_x$ and variance $\sigma^2_x$ is represented as $\mathcal{CN}(\mu_x,\sigma^2_x)$. The mathematical expectation is denoted as $\mathbb{E}\{\cdot\}$.
The real part of a complex variable $x$ is denoted as $\mathfrak{R}\{x\}$. The $M\times M$ identity matrix is denoted as ${\bf I}_M$. The Hermitian and transpose of vector/matrix ${\bf x}/{\bf X}$ are represented by ${\bf x}^H/{\bf X}^H$ and ${\bf x}^T/{\bf X}^T$, respectively. The big-O notation is represented by $\mathcal{O}$.

\section{System Model}
We consider a single-cell BC with an $N$-antenna base station (BS), serving $K$ single-antenna users, as shown in Fig. \ref{Fig-sys-model}. An $M$-element RIS is considered to aid the transmission between the BS and the users. Global perfect channel state information is assumed. The transmit signal of the BS is ${\bf x}=\sum_{k=1}^K{\bf w}_ks_k\in\mathbb{C}^{N\times 1}$, where $s_k \sim \mathcal{CN}(0,1)$ is the message intended for user $k$, which is one symbol of a codeword of length $n$, and ${\bf w}_k$ is the corresponding beamforming vector. The messages are independent and identically distributed. The power budget of the BS is denoted by $p$. The transmit power of the BS is $\mathbb{E}\{ {\bf x}^H {\bf x} \} = \sum_k {\bf w}_k^H{\bf w}_k= \sum_k \text{Tr}({\bf w}_k{\bf w}_k^H)$. The notation $\{ {\bf w} \}=\{ {{\bf {w}}_k}: \forall k\}$ is used to denote the set of all beamforming vectors.

\subsection{RIS Model}
The direct link between the users and the BS is assumed to be blocked by obstacles. Thus, the equivalent channel between the BS and user $k$ is ${\bf h}_k ( {\bf \Phi} )= {\bf f}_k {\bf \Phi} {\bf F} \in \mathbb{C}^{1\times N}$, where ${\bf F}$ is the channel matrix between the BS and the RIS, ${\bf f}_k$ is the channel vector between the RIS and user $k$, and ${\bf \Phi}$ is the matrix of RIS reflection coefficients. The feasible values of ${\bf \Phi}$ depend on the considered RIS architecture. We analyze three case studies.

\subsubsection{Locally Passive Diagonal RIS}
In this case, ${\bf \Phi}$ is a diagonal matrix, whose diagonal elements have unit modulus ($|\phi_{mm}| = 1$). Hence, the feasible set is
\begin{equation}\label{(3)}
    \mathcal{E}_{LP}=\{\phi_{mn}: |\phi_{mm}| = 1,\phi_{mn} = 0,\forall m\neq n \}.
\end{equation}

\subsubsection{Globally Passive Diagonal RIS}
In this case, ${\bf \Phi}$ is still a diagonal matrix, but the power constraint is applied globally to the whole surface. Specifically, ${\bf \Phi}$ is designed such that the output power is not greater than the input power to the RIS, but some RIS elements may amplify the incident signal \cite{fotock2023energy}. Hence, ${\bf \Phi}$ satisfies the following convex constraint:
\begin{equation}\label{(4)}
    p_{out}-p_{in}=\text{Tr}
    \left( 
    {\bf F}\mathbb{E}\{ {\bf x} {\bf x}^H \}{\bf F}^H
    ({\bf \Phi}^H{\bf \Phi}-{\bf I}_M)\right)\leq 0,
\end{equation}
where $ p_{out}$ and $p_{in}$ are the output and input power to the RIS, respectively. Hence, the feasible set is convex as:
\begin{equation}\label{(5)}
    \mathcal{E}_{D}\!=\!\{\!{\bf \Phi}\!:\!{\bf \Phi}\!=\!\text{diag},\text{Tr}
    \left( 
    {\bf F}\mathbb{E}\{ {\bf x} {\bf x}^H \}{\bf F}^H
    ({\bf \Phi}^H{\bf \Phi}\!-\!{\bf I}_M)\right)\!\leq \!0\! \}\!,\!
\end{equation}
with ``$\text{diag}$'' denoting a diagonal matrix. It is worth noting that the feasible values of ${\bf \Phi}$ depend on the transmission parameters at the BS, i.e., $\{{\bf w}\} $, which is not the case of a LP design.

\subsubsection{Globally Passive Beyond Diagonal RIS}
In a GP BD-RIS, the matrix ${\bf \Phi}$ can be non-diagonal, but needs to satisfy \eqref{(4)}. We assume that ${\bf \Phi}$ is a symmetric matrix, to ease the implementation of BD-RIS \cite{li2023beyond}. Thus, the feasible set of ${\bf \Phi}$ is
\begin{equation}\label{(6)}
    \mathcal{E}_{BD}\!\!=\!\!\{\!{\bf \Phi}\!:\! {\bf \Phi}\!=\!{\bf \Phi}^T\!,\text{Tr}\!
    \left( 
    {\bf F}\mathbb{E}\{ {\bf x} {\bf x}^H \}{\bf F}^H
    ({\bf \Phi}^H{\bf \Phi}-{\bf I}_M)\right)\!\leq\! 0 \}\!.\!\!
\end{equation}
 {Note that $\mathcal{E}_{LP}\subset\mathcal{E}_{D}\subset\mathcal{E}_{BD}$, which means that the optimal solution of GP BD-RIS cannot perform worse than any solutions obtained for LP D-RIS and GP D-RIS.}
\begin{figure}[t!]
    \centering
\includegraphics[width=.4\textwidth]{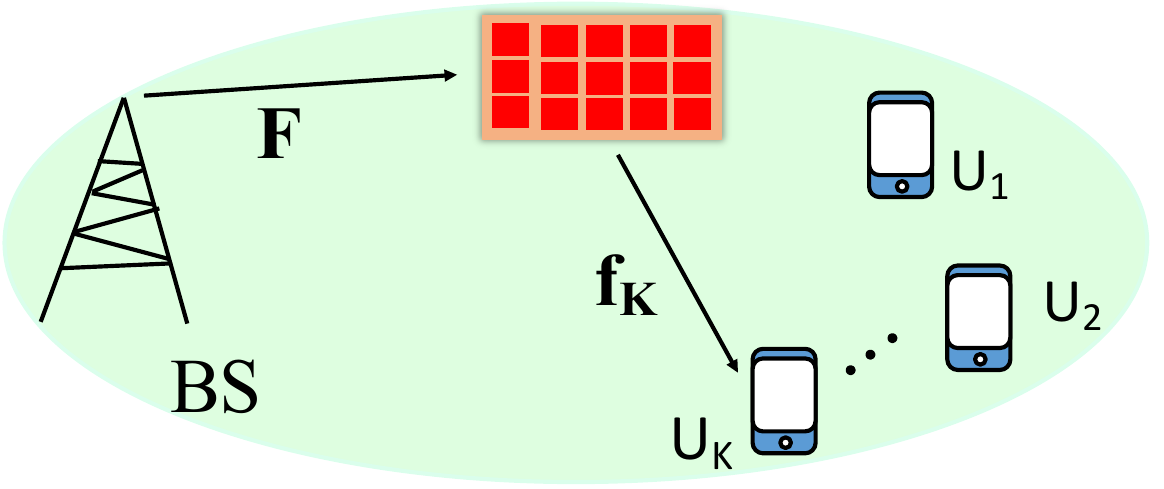} 
     \caption{ {A MISO BC assisted by an RIS.}}
	\label{Fig-sys-model}
\end{figure}

\subsection{Signal Model and Rate Definition}
The signal at user $k$ is $y_k={\bf h}_k ({\bf \Phi})\sum_{i=1}^K{\bf w}_is_i+n_k$,
where  $n_k \sim \mathcal{CN}(0,\sigma^2)$ is the additive noise. Employing TIN,
the rate of user $k$ using the normal approximation is \cite{polyanskiy2010channel, scarlett2016dispersion}
\begin{equation}\label{(8)}
    r_k=\ln (1+\gamma_k)-Q^{-1}(\epsilon)\sqrt{\frac{v_k}{n}},
\end{equation}
where $v_k$ is the channel dispersion of user $k$, $n$ is the packet length, $Q^{-1}$ is the inverse of the Q-function, $\epsilon$ is the maximum tolerable decoding error probability, and $\gamma_k$ is the signal-to-interference-plus-noise ratio (SINR) at user $k$ given by
\begin{equation}\label{(9)}
   \gamma_k=\frac{|{\bf h}_k ({\bf \Phi}){\bf w}_k|^2}{\sigma^2+\sum_{i\neq k}|{\bf h}_k ({\bf \Phi}){\bf w}_i|^2}. 
\end{equation}

An achievable value for the channel dispersion 
of Gaussian signals in interference-limited systems is \cite{scarlett2016dispersion} 
\begin{equation}\label{(10)}
   v_k= 2\frac{\gamma_k}{1+\gamma_k} . 
\end{equation}
\begin{remark}
 {In \eqref{(8)}, $\epsilon$ specifies the reliability requirement, as it denotes the maximum tolerable decoding error probability. Furthermore, the latency is specified by the finite block length $n$, as a more restrictive latency constraint corresponds to a shorter block length \cite{soleymani2024optimization, xu2022max}.}  
\end{remark}
\begin{lemma}\label{lem:Increasing}
The rate function in \eqref{(8)} is increasing in $\gamma_k$ provided that $\gamma_k\geq \bar{\gamma}=\frac{1}{2}(\sqrt{1+2c^{2}}-1)$, with $c=Q^{-1}(\epsilon)/\sqrt{n}$.
\end{lemma}
\begin{proof}
The rate function in \eqref{(8)} can be written as 
\begin{equation}
    f(\gamma)=\ln (1+\gamma)-c \sqrt{\frac{2\gamma}{(1+\gamma)}},
\end{equation}
where we have suppressed the index $k$ as it is not necessary for this proof. Computing the first-order derivative of $f$ and setting it to be non-negative yields the condition $\sqrt{2\gamma(1+\gamma)}\geq c$,
which, since both $\gamma$ and $c$ are non-negative quantities, is equivalent to the condition $\gamma\geq \bar{\gamma}$, with $\bar{\gamma}$ being the unique non-negative solution of the equation $2\gamma^{2}+2\gamma-c^{2}=0$. Through simple algebra, it is found $\bar{\gamma}=\frac{1}{2}(\sqrt{1+2c^{2}}-1)$.
\end{proof}

\subsection{Problem Statement}
We aim at deriving the achievable rate region of a MISO BC channel assisted by an RIS. The rate region is the union of all rates achievable by users, i.e.,
    $\mathcal{R}= \bigcup\nolimits_{ \{ {\bf w}\in \mathcal{W}, {\bf \Phi} \in \mathcal{E} \} } \{r_k\}_{\forall k}$,
where $\mathcal{W}$ is the feasible set of $\{{\bf w} \}$. 

Employing the rate-profile technique, the achievable rate region is obtained by solving \cite{zhang2010cooperative}
\begin{align}\label{(12)}
    \max\nolimits_{\{ \{{\bf w} \}\in \mathcal{W}, {\bf \Phi}\in \mathcal{E}, r\}} & \; r&\text{s.t.}\,\,&r_k\geq \alpha_kr, \forall k,
\end{align}
varying the coefficients $\alpha_k$ for every $\alpha_k>0$ with $\sum_k\alpha_k=1$.

Additionally, we define the SINR region as the union of all achievable SINRs, which can be obtained by solving
\begin{align}\label{(13)}
    \max\nolimits_{\{\{{\bf w} \}\in \mathcal{W}, {\bf \Phi}\in \mathcal{E}, \gamma\}} & \; \gamma&\text{s.t.}\,\,&\gamma_k\geq \lambda_k\gamma, \forall k,
\end{align}
varying the coefficients $\lambda_k$ for every $\lambda_k>0$ with $\sum_k\lambda_k=1$.

Solving \eqref{(13)} requires a much lower complexity than solving \eqref{(12)}. In the considered scenario, however, and under mild assumptions, the solution obtained by computing the SINR region is equivalent to the solution obtained by computing the rate region,  as we prove in the following lemma. 
\begin{lemma}\label{lem1}
Assume $\gamma_{k}\geq \bar{\gamma}$ for all $k=1,2,\ldots,K$. Then, every point on the boundary of the rate region corresponds to a point on the boundary of the SINR region, and vice versa.
\end{lemma}
\begin{proof}
To prove the lemma, we use a counter example showing that there cannot exist a point on the boundary of the rate region, which is not on the boundary of the SINR region and vice versa. To this end, assume that $(R_1, R_2,\cdots,R_k)$ is on the boundary of the rate region, but it is not on the boundary of the SINR region. It means that
$(\gamma_1=f^{-1}(R_1),\gamma_2=f^{-1}(R_1),\cdots,\gamma_K=f^{-1}(R_K))$ is feasible, but it is not on the boundary of the SINR region, where $f^{-1}(x)$ is the inverse of the rate function (denoted as $f(x)$ from Lemma \ref{lem:Increasing}. Indeed, Lemma \ref{lem:Increasing} ensures that $f(x)$ is a strictly increasing function whenever $\gamma_{k}\geq \bar{\gamma}$ holds.
It means that there exists at least one SINR $\gamma^\prime_k>\gamma_k$ for $1\leq k\leq K$ such that
$(\gamma_1,\gamma_2,\cdots,\gamma_k^\prime,\cdots,\gamma_K)$ is feasible, implying that the rates $(R_1,R_2,\cdots,f(\gamma_k^\prime) > R_k, \cdots,R_K)$ are achievable. Hence,
$(R_1,R_2, \cdots,R_K)$ cannot be on the boundary of the rate region, which is not correct. Similarly, it can be shown that every point on the boundary of the SINR region is associated
with a point on the boundary of the rate region.
\end{proof}
\begin{remark}
 {It is important to note that the condition $\gamma_k\geq \bar{\gamma}$ is quite realistic. Indeed, it typically holds for all practical operating conditions, as otherwise the SINR would be too low for ensuring a satisfactory system operation \cite[Lemma 2]{soleymani2023spectral}.}
\end{remark}
\begin{remark}
The results in Lemma \ref{lem1} can be extended to any multi-user MISO system using TIN, provided that the rate is a strictly increasing function of the SINR. Thus, the proposed method can be applied not only to the normal approximation, but also other approximations of the rates, as long as they are strictly increasing functions of the SINR.
\end{remark}

\section{Proposed Algorithm}
The optimization problem in \eqref{(13)} is not convex. We obtain a suboptimal solution of \eqref{(13)} by combining the majorization minimization (MM), with the alternating optimization (AO) and the generalized Dinkelbach algorithm (GDA) methods. Specifically, we first solve \eqref{(13)} with respect to $\{{\bf w}\} $, while keeping ${\bf \Phi}$ fixed. Then, we solve \eqref{(13)} with respect to ${\bf \Phi}$, while keeping $\{{\bf w}\} $ fixed. In the following, we first provide the solution for updating ${\bf \Phi}$.

\subsection{Optimization of ${\bf \Phi}$}
The problem in \eqref{(13)} when $\{{\bf w}^{(t)}\} $ is kept fixed is
\begin{subequations}\label{(15)}
\begin{align}\label{(15a)}
    \max_{{\bf \Phi}\in\mathcal{E}, \gamma}\gamma\,\,\,\;\;\text{s.t.}&\,
     \frac{|{\bf h}_k ({\bf \Phi}){\bf w}_k^{(t)}|^2}{\sigma^2+\sum_{i\neq k}|{\bf h}_k ({\bf \Phi}){\bf w}_i^{(t)}|^2}\geq 
     \lambda_k\gamma,\,\,\,\forall k,\\
     &\,\frac{|{\bf h}_k ({\bf \Phi}){\bf w}_k^{(t)}|^2}{\sigma^2+\sum_{i\neq k}|{\bf h}_k ({\bf \Phi}){\bf w}_i^{(t)}|^2}\geq\bar{\gamma},\,\,\, \forall k,\label{(15)b}
\end{align}    
\end{subequations}
where $t$ is the iteration index. The problem in \eqref{(15)}  is a multiple-ratio fractional programming (FP) problem in which the numerator and denominator are quadratic and convex functions of ${\bf \Phi}$. Also, the set $\mathcal{E}$ is not convex for LP D-RIS. 

We first handle the FP problem and then solve \eqref{(15)} for each RIS architecture.
To this end, we use the
convex-concave procedure (CCP) to approximate the numerator of $\gamma_k$ with a linear lower bound since it is a convex function. That is
\begin{multline}\label{(16)}
    |{\bf h}_k ({\bf \Phi}){\bf w}_k^{(t)}|^2\geq \hat{d}_k\left( {\bf \Phi}\right) 
     \triangleq |{\bf h}_k ({\bf \Phi}^{(t-1)}){\bf w}_k^{(t)}|^2
     \\
     \!+\!2\mathfrak{R}\!\!\left\{\! 
     {\bf h}_k ({\bf \Phi}^{(t-1)}){\bf w}_k^{(t)}\!
     \!\left(\!{\bf h}_k ({\bf \Phi}){\bf w}_k^{(t)} \!
     \!\!-\!{\bf h}_k ({\bf \Phi}^{(t-1)} ){\bf w}_k^{(t)}\!
     \right)^*\!
     \right\}\!.\!\!
\end{multline}
Substituting \eqref{(16)} in \eqref{(15)}, we obtain
\begin{subequations}\label{(17)}
\begin{align}\label{(17)a}
     \max_{{\bf \Phi}\in \mathcal{E}, \gamma} \gamma \,\,\,\;\;\text{s.t.}\,&
     \frac{\hat{d}_k\left( {\bf \Phi}\right) }{\sigma^2+\sum_{i\neq k}|{\bf h}_k ({\bf \Phi}){\bf w}_i^{(t)}|^2}\geq \lambda_k\gamma, \forall k,\\
     &\frac{\hat{d}_k\left( {\bf \Phi}\right) }{\sigma^2+\sum_{i\neq k}|{\bf h}_k ({\bf \Phi}){\bf w}_i^{(t)}|^2}\geq \bar{\gamma}, \forall k,\label{(17)b}
\end{align}    
\end{subequations}
which is a non-convex problem, but it can be optimally solved by using the GDA when $ \mathcal{E}$ is a convex set. Indeed, \eqref{(17)b} can be rewritten as the convex constraint
\begin{equation}\label{Eq:NewConstraint}
    \hat{d}_k\left( {\bf \Phi}\right)-\bar{\gamma}(\sigma^2+\sum\nolimits_{i\neq k}|{\bf h}_k ({\bf \Phi}){\bf w}_i^{(t)}|^2)\geq 0, \,\,\,\forall k.
\end{equation}
Thus, we can compute a solution of
\eqref{(17)} when $ \mathcal{E}$ is convex, by iteratively solving the problem
\begin{subequations}\label{(18)}
\begin{align}
     \!\!\!\max_{{\bf \Phi}\in\mathcal{E}}\, &\min\!\!\left\{\!\!
     {\hat{d}_k\!\left( {\bf \Phi}\right) }\!-\!\mu^{(t,m)}\lambda_k({\sigma^2\!+\!\!\sum_{i\neq k}|{\bf h}_k ({\bf \Phi}){\bf w}_i^{(t)}|^2})\!\!
     \right\}\!\\
     &\!\!\!\text{s.t}\;\hat{d}_k\!\left( {\bf \Phi}\right)\!-\!\bar{\gamma}(\sigma^2\!+\!\sum\nolimits_{i\neq k}\!|{\bf h}_k ({\bf \Phi}){\bf w}_i^{(t)}|^2)\geq 0,\;\forall\, k\label{(18b)}
\end{align}
\end{subequations}
and updating $\mu^{(t,m)}$ as
\begin{equation}\label{(19)}
    \mu^{(t,m)}=\min_k\! \left\{\! \frac{\hat{d}_k\left( {\bf \Phi}^{(t,m)}\right) }{\lambda_k\!
    \left(\!\sigma^2\!+\sum_{i\neq k}\left| {\bf h}_k \left( {\bf \Phi}^{(t,m)}\right){\bf w}_i^{(t)} \right|^2
    \right)}\! \right\}\!,\!\!
\end{equation}
where ${\bf \Phi}^{(t,m)}$ is the initial point at the $m$-th iteration of the GDA, 
which is the solution of \eqref{(18)} at the previous iteration.

\subsubsection{GP BD-RIS}
To compute ${\bf \Phi}$ for GP BD-RIS, we iteratively solve the convex problem
\begin{subequations}\label{(20)}
\begin{align}
   \!\!\!  \max_{{\bf \Phi}} \,&\min\!\left\{\!
     {\hat{d}_k\left( {\bf \Phi}\right) }\!-\!\mu^{(t,m)}\lambda_{k}({\sigma^2\!+\!\sum_{i\neq k}|{\bf h}_k ({\bf \Phi}){\bf w}_i^{(t)}|^2})\!
     \right\}
     \\
     \text{s.t.}\,\,&
     \label{(15b)}
     \sum_k\!\text{Tr}\!
    \left( \!
    {\bf F}{\bf w}_k^{(t)}{\bf w}_k^{(t)^H}{\bf F}^H
    ({\bf \Phi}^H{\bf \Phi}\!-\!{\bf I}_M)\!\right)\!\leq\! 0,
     \\
     \label{(15c)}
     &\eqref{Eq:NewConstraint},\;\;{\bf \Phi}={\bf \Phi}^T,
\end{align}    
\end{subequations}
by updating $\mu^{(t,m)}$ according to \eqref{(19)}.

\subsubsection{GP D-RIS}
To compute ${\bf \Phi}$ for GP D-RIS, we solve the convex problem
\begin{subequations}\label{(21)}
\begin{align}
    \!\!\! \max_{{\bf \Phi}}\, &\min\!\left\{\!
     {\hat{d}_k\left( {\bf \Phi}\right) }\!-\!\mu^{(t,m)}\lambda_{k}({\sigma^2\!+\!\sum_{i\neq k}|{\bf h}_k ({\bf \Phi}){\bf w}_i^{(t)}|^2})\!
     \right\}
     \\
     \text{s.t.}\,\,&
     \eqref{Eq:NewConstraint}, \, \eqref{(15b)},\phi_{ij}=0, \,\,\forall i\neq j, 
     \label{(21b)}
\end{align}
\end{subequations}
by updating $\mu^{(t,m)}$ according to \eqref{(19)}.
The algorithm for GP D-RIS is similar to that of GP BD-RIS. The only difference is substituting the symmetry constraint in \eqref{(15c)} with the convex constraint  $\phi_{ij}=0$ for all $i\neq j$ to make ${\bf \Phi}$ diagonal. 
 {Note that the proposed solutions for GP D-RIS and GP BD-RIS converge to a stationary point of \eqref{(13)} since the proposed algorithms for these architectures fulfill the MM convergence conditions \cite{sun2017majorization}.}

\subsubsection{LP D-RIS}
We employ again the lower bound in \eqref{(16)} and  the GDA, which yields
\begin{subequations}\label{(22)}
\begin{align}
    \!\!\! \max_{{\bf \Phi}}\, &\min\!\left\{\!
     {\hat{d}_k\left( {\bf \Phi}\right) }\!-\!\mu^{(t,m)}\lambda_{k}({\sigma^2\!+\!\sum_{i\neq k}|{\bf h}_k ({\bf \Phi}){\bf w}_i^{(t)}|^2})\!
     \right\}
     \\
     \text{s.t.}\,\,&
     \eqref{Eq:NewConstraint},\,|\phi_{ii}|=1,\phi_{ij}=0, \,\,\forall i\neq j,\; 
\end{align}
\end{subequations}
with $\mu^{(t,m)}$ given in \eqref{(19)}.
The problem in \eqref{(22)} is not convex since $|\phi_{ii}|=1$ is not a convex set. To tackle it, we convexify the constraint $|\phi_{ii}|=1$ as detailed in \cite[Sec. IV.B-1]{soleymani2023spectral}.  {The algorithm for LP D-RIS converges since it generates a sequence of non-decreasing values of $\min_k\{\gamma_k\}$ \cite{soleymani2023spectral}.}

\subsection{Optimization of $\{{\bf w} \}$}\label{sec-iii-b}
To compute $\{{\bf w} \}$, we employ the linear lower bound in \eqref{(16)} for the numerator of the SINRs and use the GDA to obtain the global optimal solution of the corresponding surrogate optimization problem. Specifically, we obtain $\{{\bf w} \}$ by solving  
\begin{subequations}\label{(22)}
\begin{align}
     \max_{{\bf w}}\,& \min\!\!\left\{\!\!
     {\tilde{d}_k\left( {\bf w}\right) }\!-\!\mu^{(t,m)}\lambda_{k}({\sigma^2\!+\!\sum_{i\neq k}|{\bf h}_k ({\bf \Phi}^{(t)}){\bf w}_i|^2})\!\!
     \right\}\!\!
     \\
     \text{s.t.}\,\,&
     \,\tilde{d}_k\left( {\bf w}\right)\!-\!\bar{\gamma}(\sigma^2\!+\!\sum_{i\neq k}|{\bf h}_k ({\bf \Phi}){\bf w}_i^{(t)}|^2)\!\geq\! 0, \,\,\, \forall k,
     \\
     &\sum_k \!\|{\bf w}_k\|^{2}\!\leq\! p,
\end{align}
\end{subequations}
where $\mu^{(t,m)}$ is given in \eqref{(19)}, and  
\begin{multline}\label{(23)}
\tilde{d}_k\left( {\bf w}\right) 
     \!\triangleq \!|{\bf h}_k ({\bf \Phi}^{(t-1)}){\bf w}_k^{(t-1)}|^2
     \!+\!2\mathfrak{R}\!\left\{\! 
     {\bf h}_k (\!{\bf \Phi}^{(t-1)}){\bf w}_k^{(t-1)}\!
     \right.
     \\
     \left.
     \left({\bf h}_k (\!{\bf \Phi}^{(t-1)}){\bf w}_k 
     -{\bf h}_k ({\bf \Phi}^{(t-1)} ){\bf w}_k^{(t-1)}
     \right)^*
     \right\}.
\end{multline}

\subsection{Computational Complexity}
We estimate the computational complexity in terms of the number of complex multiplications to calculate ${\bf \Phi}$ for GP BD-RIS and GP D-RIS. A similar analysis can be conducted for LP D-RIS and $\{{\bf w}\}$. To obtain ${\bf \Phi}$, \eqref{(21)} and \eqref{(22)} need to be solved numerically for GP BD-RIS and GP D-RIS, respectively. To this end, we note that the number of Newton iterations to numerically solve a convex problem is proportional to the square root of the number of constraints \cite{boyd2004convex}. 

As for $\mathcal{E}_D$ and $\mathcal{E}_{BD}$, two constraints for the feasible set of ${\bf \Phi}$ exist. In addition, there are $K$ constraints for the SINR of each user. Hence, the number of Newton iterations grows with $\sqrt{K}$. To solve \eqref{(20)} or \eqref{(21)} at each Newton step, $K$ channels need to be computed, which approximately requires $KM^2N$ and $KMN$  multiplications for BD-RIS and D-RIS, respectively.
To compute the objective function in \eqref{(20)} and \eqref{(21)}, the number of multiplications grows with $K^2N$. Additionally, to compute the constraint \eqref{(15b)}, the number of multiplications grows with $KM^2(M+N)$ and $KM^2N$ for GP BD-RIS and GP D-RIS, respectively. Therefore, the computational complexity of one iteration of the GDA for GP BD-RIS and GP D-RIS is of order of $\mathcal{O}\left(K\sqrt{K}M^2(M+N)\right)$ and $\mathcal{O}\left(K\sqrt{K}M^2N\right)$, respectively. In conclusion, the complexity of GP BD-RIS and GP D-RIS increases with $M^3$ and $M^2$, respectively.

\section{Numerical Results}
Monte Carlo simulations are utilized to obtain the numerical results. The channel between the BS and the RIS follows a Rician distribution with Rician factor equal to $10$ \cite[Eqs. (60)-(62)]{soleymani2022improper}. The channels between the RIS and the users are in non-line of sight (NLOS) and follow a Rayleigh distribution. The heights of the BS and RIS are equal to $25$ meters, while the height of the users is $1.5$ meters. The BS is located at the origin, and the distance between the BS and the RIS is $1$ meter. The users are randomly located in a square area with side equal to $20$ meters, and the distance between the center of the square region and the BS is $130$ meters. For comparison, a scenario without the RIS (denoted by {\bf No-RIS}) is considered. In this case, the channels between the BS and the users are NLOS and follow a Rayleigh distribution. The simulation setup is detailed in the captions of each figure.  {Also, a setup where the reflection coefficients of the RIS are randomly distributed (denoted by {\bf RIS-Rand}) is illustrated as well.}

\begin{figure}[t]
    \centering
    \begin{subfigure}[t]{0.24\textwidth}
        \centering
           \includegraphics[width=\textwidth]{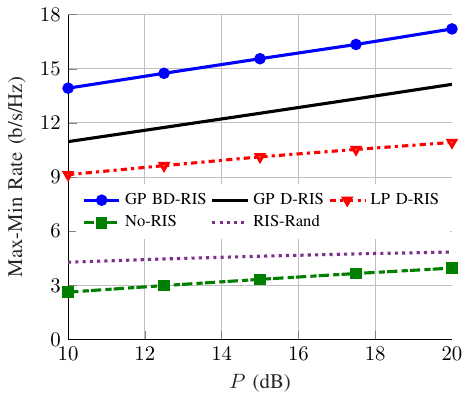}
        \caption{$K=4$.}
    \end{subfigure}
\begin{subfigure}[t]{0.24\textwidth}
        \centering
       \includegraphics[width=\textwidth]{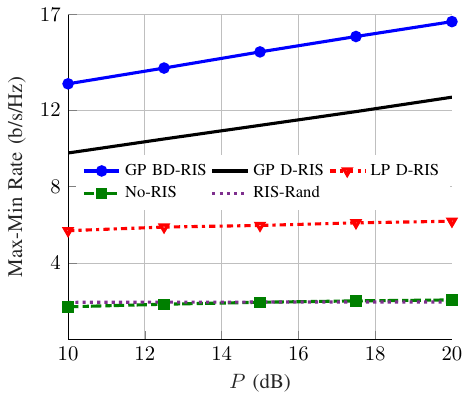}
        \caption{$K=5$.}
    \end{subfigure}%
    \caption{Max-min rate versus $P$ (Setup: $N=6$, $n=256$ bits, $\epsilon=10^{-5}$, $M=20$).}
	\label{Fig1} 
\end{figure}

Fig. \ref{Fig1} shows the average value of the max-min rate as a function of the transmit power $P$. Specifically, we set $\alpha_k=1/K$ for all $k$. 
We see that GP BD-RIS outperforms the other RIS architectures in terms of max-min rate. The price to pay is usually a higher implementation complexity, computational complexity, and reduced energy efficiency \cite{soleymani2024energy}. 
Moreover, when $K=5$, the gap between the LP and GP architectures  (D or BD) is higher. Indeed, in this case, the performance of LP D-RIS is mainly limited by interference, and its max-min rate does not noticeably vary with $P$. However, GP D-RIS and GP BD-RIS can more effectively manage interference than LP D-RIS when $K=5$.
  {As expected, we observe that Rand-RIS does not provide any considerable gain, since its elements are not optimized.}

\begin{figure}[t]
    \centering
    \begin{subfigure}[t]{0.24\textwidth}
        \centering
           \includegraphics[width=\textwidth]{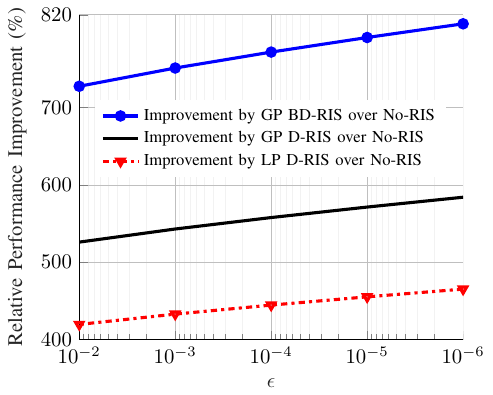}
        \caption{Gain versus $\epsilon$ ($n=256$ bits)}
    \end{subfigure}
\begin{subfigure}[t]{0.24\textwidth}
        \centering
       \includegraphics[width=\textwidth]{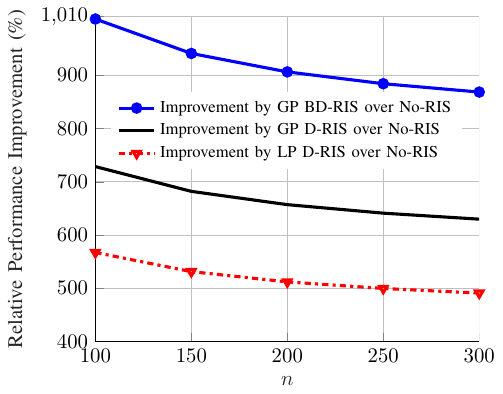}
        \caption{Gain versus $n$ ($\epsilon=10^{-5}$)}
    \end{subfigure}%
    \caption{Gain provided by different RIS architectures (Setup: $N=3$, $K=3$, $M=20$, $P=10$ dB).}
	\label{Fig23} 
\end{figure}

Fig. \ref{Fig23} illustrates the gain provided by different RIS architectures, as compared with No-RIS, in terms of average max-min rate, and as a function of $\epsilon$ and $n$. This figure shows that the considered RIS architectures substantially increase the average max-min rate with respect to the No-RIS deployment.  {Additionally, the gain of each RIS architecture increases when the reliability and latency constraints are more stringent, i.e., when the packets are shorter and when $\epsilon$ is lower.} In addition, the average max-min rate of GP BD-RIS is approximately 33\% and 66\% higher than that of GP D-RIS and LP D-RIS, respectively.

\begin{figure}[t]
    \centering
    \begin{subfigure}[t]{0.24\textwidth}
        \centering
           \includegraphics[width=.85\textwidth]{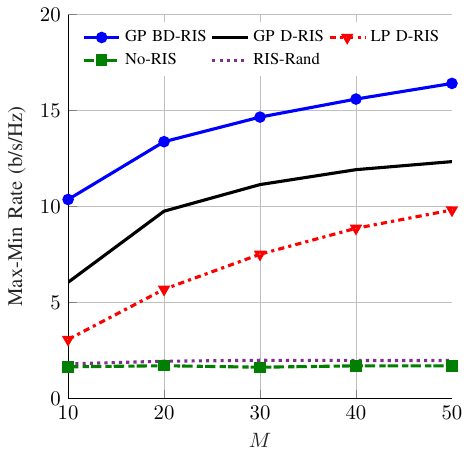}
        \caption{Average max-min rate}
    \end{subfigure}
\begin{subfigure}[t]{0.24\textwidth}
        \centering
       \includegraphics[width=\textwidth]{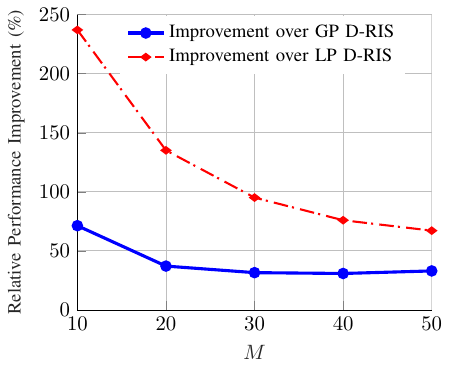}
        \caption{ {Improvement provided by GP BD-RIS}}
    \end{subfigure}%
    \caption{Impact of $M$. (Setup: $N=6$, $K=5$, $n=256$ bits, $\epsilon=10^{-5}$, $P=10$ dB).}
	\label{Fig4} 
\end{figure}

Fig. \ref{Fig4} illustrates the average max-min rate as a function of the number $M$ of RIS elements.  {We see in Fig. \ref{Fig4}a that the average max-min rate of RIS-aided systems increases with $M$ only if the RIS elements are optimized according to the proposed schemes. Interestingly, however, the gain provided by GP BD-RIS over GP D-RIS and LP D-RIS decreases as $M$ increases, as depicted in Fig. \ref{Fig4}b.} Recalling that the number of electronic circuits for D-RIS and fully connected BD-RIS is $M$ and $M(M-1)/2$, respectively, the benefits of BD-RIS needs to account for the implementation complexity and power consumption of the required electronics as well.

\section{Conclusion}
In this paper, we analyzed the rate region of RIS-aided BC channels by considering FBL coding and assuming TIN-based decoding.  {We proved that the rate region can be obtained from the SINR region, provided that the rate increases monotonically with the SINR.}  {We compared three RIS architectures (LP D-RIS, GP D-RIS, and GP BD-RIS), and have shown that  GP BD-RIS, which includes LP D-RIS and GP D-RIS as special cases, provides the highest average max-min rate, at the cost of increasing the implementation and computational complexities.} The gain provided by GP BD-RIS reduces, in addition as the number of RIS elements increases.

\bibliographystyle{IEEEtran}
\bibliography{ref2}

\begin{thebibliography}{10}
\providecommand{\url}[1]{#1}
\csname url@samestyle\endcsname
\providecommand{\newblock}{\relax}
\providecommand{\bibinfo}[2]{#2}
\providecommand{\BIBentrySTDinterwordspacing}{\spaceskip=0pt\relax}
\providecommand{\BIBentryALTinterwordstretchfactor}{4}
\providecommand{\BIBentryALTinterwordspacing}{\spaceskip=\fontdimen2\font plus
\BIBentryALTinterwordstretchfactor\fontdimen3\font minus \fontdimen4\font\relax}
\providecommand{\BIBforeignlanguage}[2]{{%
\expandafter\ifx\csname l@#1\endcsname\relax
\typeout{** WARNING: IEEEtran.bst: No hyphenation pattern has been}%
\typeout{** loaded for the language `#1'. Using the pattern for}%
\typeout{** the default language instead.}%
\else
\language=\csname l@#1\endcsname
\fi
#2}}
\providecommand{\BIBdecl}{\relax}
\BIBdecl

\bibitem{di2020smart}
M.~Di~Renzo, A.~Zappone, M.~Debbah, M.-S. Alouini, C.~Yuen, J.~De~Rosny, and S.~Tretyakov, ``Smart radio environments empowered by reconfigurable intelligent surfaces: How it works, state of research, and the road ahead,'' \emph{IEEE J. Sel. Areas Commun.}, vol.~38, no.~11, pp. 2450--2525, 2020.

\bibitem{wu2021intelligent}
Q.~Wu, S.~Zhang, B.~Zheng, C.~You, and R.~Zhang, ``Intelligent reflecting surface aided wireless communications: A tutorial,'' \emph{IEEE Trans. Commun.}, vol.~69, no.~5, pp. 3313--3351, 2021.

\bibitem{fotock2023energy}
R.~K. Fotock, A.~Zappone, and M.~Di~Renzo, ``Energy efficiency optimization in {RIS}-aided wireless networks: Active versus nearly-passive {RIS} with global reflection constraints,'' \emph{IEEE Trans. Commun.}, vol.~72, no.~1, pp. 257--272, 2024.

\bibitem{li2023beyond}
H.~Li, S.~Shen, and B.~Clerckx, ``Beyond diagonal reconfigurable intelligent surfaces: From transmitting and reflecting modes to single-, group-, and fully-connected architectures,'' \emph{IEEE Trans. Wireless Commun.}, vol.~22, no.~4, pp. 2311--2324, 2023.

\bibitem{BartoliADDR23}
G.~Bartoli, A.~Abrardo, N.~Decarli, D.~Dardari, and M.~Di~Renzo, ``Spatial multiplexing in near field {MIMO} channels with reconfigurable intelligent surfaces,'' \emph{{IET} Signal Process.}, vol.~17, 2023.

\bibitem{soleymani2024energy}
M.~Soleymani \emph{et~al.}, ``Energy efficiency comparison of {RIS} architectures in {MISO} broadcast channels,'' in \emph{IEEE Int. Workshop Signal Process. Adv. Wireless Commun. (SPAWC)}.\hskip 1em plus 0.5em minus 0.4em\relax IEEE, 2024.

\bibitem{wang2023road}
C.-X. Wang \emph{et~al.}, ``On the road to {6G}: Visions, requirements, key technologies and testbeds,'' \emph{IEEE Commun. Surv. Tutor.}, vol.~25, no.~2, pp. 905--974, 2023.

\bibitem{soleymani2023spectral}
M.~Soleymani, I.~Santamaria, and E.~Jorswieck, ``Spectral and energy efficiency maximization of {MISO} {STAR-RIS}-assisted {URLLC} systems,'' \emph{IEEE Access}, vol.~11, pp. 70\,833--70\,852, 2023.

\bibitem{almekhlafi2021joint}
M.~Almekhlafi, M.~A. Arfaoui, M.~Elhattab, C.~Assi, and A.~Ghrayeb, ``Joint resource allocation and phase shift optimization for {RIS}-aided {eMBB/URLLC} traffic multiplexing,'' \emph{IEEE Trans. Commun.}, vol.~70, no.~2, pp. 1304--1319, 2022.

\bibitem{soleymani2023optimization}
M.~Soleymani, I.~Santamaria, E.~Jorswieck, and B.~Clerckx, ``Optimization of rate-splitting multiple access in beyond diagonal {RIS}-assisted {URLLC} systems,'' \emph{IEEE Trans. Wireless Commun.}, vol.~23, no.~5, pp. 5063--5078, 2024.

\bibitem{vu2022intelligent}
T.-H. Vu, T.-V. Nguyen, D.~B. da~Costa, and S.~Kim, ``Intelligent reflecting surface-aided short-packet non-orthogonal multiple access systems,'' \emph{IEEE Trans. Veh. Technol.}, vol.~71, no.~4, pp. 4500--4505, 2022.

\bibitem{soleymani2024optimization}
M.~Soleymani \emph{et~al.}, ``Optimization of the downlink spectral-and energy-efficiency of {RIS}-aided multi-user {URLLC} {MIMO} systems,'' \emph{IEEE Trans. Commun.}, 2024, doi: 10.1109/TCOMM.2024.3480988.

\bibitem{polyanskiy2010channel}
Y.~Polyanskiy, H.~V. Poor, and S.~Verd{\'u}, ``Channel coding rate in the finite blocklength regime,'' \emph{IEEE Trans. Inf. Theory}, vol.~56, no.~5, pp. 2307--2359, 2010.

\bibitem{scarlett2016dispersion}
J.~Scarlett, V.~Y. Tan, and G.~Durisi, ``The dispersion of nearest-neighbor decoding for additive non-{G}aussian channels,'' \emph{IEEE Trans. Inf. Theory}, vol.~63, no.~1, pp. 81--92, 2016.

\bibitem{xu2022max}
Y.~Xu, Y.~Mao, O.~Dizdar, and B.~Clerckx, ``Max-min fairness of rate-splitting multiple access with finite blocklength communications,'' \emph{IEEE Trans. Veh. Technol.}, vol.~72, no.~5, pp. 6816--6821, 2023.

\bibitem{zhang2010cooperative}
R.~Zhang and S.~Cui, ``Cooperative interference management with {MISO} beamforming,'' \emph{IEEE Trans. Signal Process.}, vol.~58, no.~10, pp. 5450--5458, 2010.

\bibitem{sun2017majorization}
Y.~Sun, P.~Babu, and D.~P. Palomar, ``Majorization-minimization algorithms in signal processing, communications, and machine learning,'' \emph{IEEE Trans. Signal Process.}, vol.~65, no.~3, pp. 794--816, 2017.

\bibitem{boyd2004convex}
S.~Boyd and L.~Vandenberghe, \emph{Convex {O}ptimization}.\hskip 1em plus 0.5em minus 0.4em\relax Cambridge University Press, 2004.

\bibitem{soleymani2022improper}
M.~Soleymani, I.~Santamaria, and P.~J. Schreier, ``Improper signaling for multicell {MIMO} {RIS}-assisted broadcast channels with {I/Q} imbalance,'' \emph{IEEE Trans. Green Commun. Netw.}, vol.~6, no.~2, pp. 723--738, 2022.

\end{thebibliography}
\end{document}